\documentclass[11pt]{article}

\usepackage{a4wide}
\usepackage{amsmath}
\usepackage{amssymb}
\usepackage{tikz}
\usepackage{amsthm, mathtools}

\usepackage[colorlinks=true, linkcolor=blue, citecolor=blue]{hyperref}
\usepackage[nameinlink, noabbrev]{cleveref}
\usepackage{float}
\usepackage{mathrsfs}

\theoremstyle{plain} % italic body, bold head
\newtheorem{theorem}{Theorem}[section]      % numbered within sections
\newtheorem{corollary}[theorem]{Corollary}

\theoremstyle{definition} % upright body, bold head

\theoremstyle{remark} % upright body, italic head

\newtheorem*{theorem*}{Theorem}
\newtheorem*{lemma*}{Lemma}

\crefname{theorem}{Theorem}{Theorems}
\crefname{lemma}{Lemma}{Lemmas}
\crefname{proposition}{Proposition}{Propositions}
\crefname{corollary}{Corollary}{Corollaries}
\crefname{definition}{Definition}{Definitions}
\crefname{remark}{Remark}{Remarks}
\crefname{assumption}{Assumption}{Assumptions}
\crefname{example}{Example}{Examples}

\definecolor{darkgreen}{RGB}{0,175,20}

\title{Consistent pricing of bivariate interest rate exotics\\
via constrained Schr\"odinger optimal transport}
\author{Patrick Roome}
\date{June 2026}
\begin{document}
\maketitle

\begin{center}
\begin{minipage}{0.9\linewidth}
\small\bfseries
The statements, views and opinions expressed in this article are my own and do not necessarily
reflect those of my employer, JPMorgan Chase \& Co., its affiliates, other employees or clients.
\end{minipage}
\end{center}

\vspace{1em}

\begin{abstract}
We develop a modeling framework for pricing bivariate interest rate exotic derivatives that maintains consistency across three interconnected markets: CMS spread options and the two underlying CMS option markets that define the spread.
Our approach also enables the computation of no-arbitrage bounds for exotic derivatives given observable market prices in the spread option and underlying CMS option markets.
The method relies on solving the dual Lagrangian of a constrained version of the Shr\"odinger optimal transport problem and we demonstrate the practical applicability of our framework through concrete numerical examples that illustrate both the pricing methodology and the computation of no-arbitrage bounds. 
The approach offers a robust tool for pricing complex interest rate derivatives while ensuring consistency with liquid market instruments.
\end{abstract}

\section{Introduction}

Bivariate interest rate exotics are vanilla options with non-linear payoffs depending on two CMS rates. Any pricing model for these products must reproduce both CMS marginal distributions and the spread distribution implied by CMS spread option prices. The most direct approach is to specify a copula, which by construction matches the two marginal distributions. However, specifying a copula to accurately replicate the spread option market is non-trivial. While solutions have been proposed (see the power copulas in~\cite{Andersen2010}), achieving precise fits to the spread volatility smile remains challenging and often cumbersome.

In the context of FX options, Austing~\cite{Austing2011} proposed a candidate joint distribution that reproduces both the FX cross smile and the two underlying FX marginal distributions. This methodology translates naturally to interest rate markets (in normal rather than lognormal space), yielding a joint distribution consistent with the two CMS marginals and the spread distribution implied by spread option prices. While appealing, this approach has several drawbacks. 
Existence of the candidate joint distribution is not necessarily guaranteed, even when a valid joint distribution matching all three markets exists.
The approach also provides limited intuition about the resulting joint distribution relative to the full space of admissible distributions.

Piterbarg~\cite{Piterbarg2011} provides necessary and sufficient conditions for the existence of a joint distribution consistent with all three markets. He also derives no-arbitrage bounds for exotic prices by solving a linear program (and subsequently a quadratic program) subject to marginal and spread constraints. A drawback of this approach is its computational cost: solving the linear program scales cubically in the number of discretization points, which becomes prohibitive for large-scale problems.

In this paper, we provide a unified framework to construct a joint distribution consistent with all three markets and to compute no-arbitrage bounds for exotics given observable prices in the CMS and spread option markets. The joint distribution is anchored to a financially intuitive prior distribution\footnote{We use a Gaussian copula to construct the prior joint density but any reasonable joint density can be used in the construction.}, providing interpretability. Moreover, computing no-arbitrage bounds is faster than directly solving the linear (or quadratic) program as in~\cite{Piterbarg2011}. Our methodology solves the dual Lagrangian of a constrained Schr\"odinger bridge problem and is closest in spirit to~\cite{Guyon2024} and~\cite{GuyonBourgey2024}.

The paper is organized as follows. Section~\ref{sec: bridge} formally states the problem, establishes existence of solutions, and discusses the standard Sinkhorn approach. Section~\ref{sec: dual} derives the dual formulation and explains our preference for the dual problem. Section~\ref{sec:impl} details the implementation and solution, and Section~\ref{sec:numerics} applies the methodology to practical examples.

\section{Shr\"odinger bridge problem}\label{sec: bridge}

We are given a probability space $(\Omega,\mathcal{F},\mathbb{P})$ and two discrete random variables $X:\Omega\to\mathbb{R}$ and $Y:\Omega\to\mathbb{R}$ defined on that space with measures $\mu_{X}:=\sum_{i=1}^{n} \delta_{x_i}a_i$ and $\mu_{Y}:=\sum_{j=1}^{m} \delta_{y_j}b_j$ respectively.
Further, we suppose that we are given the measure $\nu_{Z}:=\sum_{k=1}^{l} \delta_{z_k}c_k$ of their spread $Z:=X-Y$ for some $ m+n-1\leq l \leq mn$.

Next we define $\mathcal{M}(\mu_{X},\mu_{Y},\nu_{Z})$ as the set of all joint (discrete) probability measures $P$ on $\mathbb{R}_{+}^{n,m}$ such that 
$$
X\sim \mu_{X}, \qquad Y\sim \mu_{Y} \quad \text{ and }\quad X-Y\sim \nu_{Z}.
$$
Necessary and sufficient conditions for the existence of a joint distribution satisfying these conditions is given in~\cite{Piterbarg2011} and we assume throughout that $\mathcal{M} \neq \emptyset$.
The set $\mathcal{M}$ can be written as $\mathcal{M}=\mathscr{C}_1\cap \mathscr{C}_2\cap \mathscr{C}_3 \cap \mathbb{R}_{+}^{n,m}$ where
\begin{align}
\mathscr{C}_1 &:= \Bigl\{P\in \mathbb{R}^{n \times m} : \sum_{j=1}^{m} P_{ij} = a_i,\; i=1,\ldots,n\Bigr\}, \nonumber \\[4pt]
\mathscr{C}_2 &:= \Bigl\{P\in \mathbb{R}^{n \times m} : \sum_{i=1}^{n} P_{ij} = b_j,\; j=1,\ldots,m\Bigr\}, \label{eq:constraint} \\[4pt] \nonumber
\mathscr{C}_3 &:= \Bigl\{P\in \mathbb{R}^{n \times m} : \sum_{(i,j)\in\mathscr{D}_k} P_{ij} = c_k,\; k=1,\ldots,l\Bigr\},
\end{align}
with 
\begin{equation}\label{eq:spreadset}
\mathscr{D}_{k} := \{(i,j) : x_i - y_j = z_k\}.
\end{equation}
For a reference joint probability $Q$ we define the Kullback--Leibler divergence between couplings as
$$
  \mathrm{KL}(P \mid\mid Q)
 :=
  \sum_{i,j} P_{i,j}\,\log\!\left(\frac{P_{i,j}}{Q_{i,j}}\right).
$$
Here if $P$ is not absolutely continuous w.r.t.$Q$, then we define $\mathrm{KL}(P \mid\mid Q)=+\infty$.
For a given payoff $C\geq0$ of $X$ and $Y$ the problem we aim to solve is 
\begin{equation}\label{eq:mainProb}
L^{\varepsilon,\beta}_{C}(a,b,c):=\operatorname*{min}_{P \in \mathscr{C}_1\cap \mathscr{C}_2\cap \mathscr{C}_3} \beta \sum_{i,j} P_{i,j} C_{i,j}+\varepsilon \mathrm{KL}(P \mid\mid Q).
\end{equation}
Note that the Kullback--Leibler divergence already embeds the positivity constraint for $P$. 
We are interested in the cases $\beta=\pm1, \varepsilon>0$ (which corresponds to computing bounds on the exotic with payoff $C$) and $\beta=0,\varepsilon=1$ (which corresponds to a pricing model independent of $C$).
This problem can easily be re-cast as a (static) Shr\"odinger bridge problem:
\begin{equation}\label{eq:ShrodProb}
 \operatorname*{min}_{P \in \mathscr{C}_1\cap \mathscr{C}_2\cap \mathscr{C}_3}  \mathrm{KL}(P \mid\mid Q e^{-\beta C/\varepsilon}).
\end{equation}
For a given convex set $\mathscr{C}\subset \mathbb{R}^{n\times m}$ the projection according to the Kullback--Leibler divergence is defined as
$$
P^{\textrm{KL}}_{\mathscr{C}}(\xi):=\operatorname*{arg\,min}_{P \in \mathscr{C}} \mathrm{KL}(P \mid\mid \xi).
$$
Since $\mathscr{C}_1$, $\mathscr{C}_2$ and $\mathscr{C}_3$ are all affine subspaces, the problem~\eqref{eq:ShrodProb} has a unique solution and can be solved by using iterative Bregman projections (see~\cite{BenamouCarlierCuturiNennaPeyre2015}).
More specifically, we initialise $P^{(0)}=Q e^{-\beta C/\varepsilon}$ and compute
\begin{equation}\label{eq:Bregman}
\forall n \in \mathbb{N}, \quad P^{(n)}:=P^{\textrm{KL}}_{\mathscr{C}_n}(P^{(n-1)}), 
\end{equation}
where we extend the indexing of the sets by $3$-periodicity, so that they satisfy
$$
\forall n \in \mathbb{N},\quad \mathscr{C}_{n+3}=\mathscr{C}_{n}.
$$
Then $ P^{(n)}$ converges to the unique solution of~\eqref{eq:ShrodProb}:
$$
P^{(n)} \to P^{\textrm{KL}}_{\mathscr{C}}(Q e^{-\beta C/\varepsilon})\quad \text{as }n\to\infty.
$$
Applying Bregman iterative projections~\eqref{eq:Bregman} to this splitting yields the well-known IPFP/Sinkhorn algorithm. Numerous papers discuss fast and robust implementations (see for example~\cite{Cuturi2013} and~\cite{PeyreCuturi2018}).

\section{Dual Lagrangian formulation}\label{sec: dual}

Typically we will be given three continuous distributions which we would need to discretise before applying the Sinkhorn methodology above. The discretisation structure of $X$ and $Y$ determines the cardinality of $\mathrm{supp}(\nu_{Z})$, which can range from $m+n-1$ to $mn$ elements depending on the grid alignment. Ideally, we would discretise using points from an efficient quadrature scheme. However, this requires first identifying $\mathrm{supp}(\nu_{Z})$, which is computationally undesirable. A common workaround is to define $\mathrm{supp}(\mu_X)$ and $\mathrm{supp}(\mu_Y)$ as arithmetic progressions with identical step sizes. While this ensures $\mathrm{supp}(\nu_{Z})$ is trivial to construct and always contains exactly $m+n-1$ elements, it sacrifices the connection to efficient and robust quadrature rules.

More fundamentally, our objective is to accurately reprice market option prices. The primal formulation above does not provide direct control over model-implied option prices. Although the Sinkhorn algorithm ensures that model-implied densities converge to market-implied densities, it remains unclear how small density discrepancies propagate to option pricing errors. This issue is particularly acute in interest rate markets, where long-dated maturities, convexity adjustments, and right-skewed densities amplify the sensitivity of option prices to subtle density variations.

Except in simple cases, the primal methodology fails to achieve the requisite accuracy for matching market option prices. We therefore reformulate the problem via its dual, which provides explicit control over option pricing errors and can be solved efficiently using standard gradient-based optimisation techniques:

\begin{theorem}\label{thm:dual}
The following duality results holds (where $L^{\varepsilon,\beta}_{C}(a,b,c)$ is defined in~\eqref{eq:mainProb}):
\begin{eqnarray*}
L^{\varepsilon,\beta}_{C}(a,b,c)=\operatorname*{max}_{f\in\mathbb{R}^{n},g\in\mathbb{R}^{m},h\in\mathbb{R}^{l}}
\sum_{i}f_i a_i&+&\sum_{j}g_j b_j + \sum_{k}h_k c_k  \\
&-&\varepsilon \sum_{i,j} Q_{i,j}e^{(f_i+g_j+h_{k(i,j)}-\beta C_{i,j})/\varepsilon} +\varepsilon,
\end{eqnarray*}
where $k(i,j)$ is such that $(i,j)\in \mathscr{D}_{k}$ (defined in~\eqref{eq:spreadset}).
\end{theorem}
\begin{proof}
Introducing the dual variables $f \in \mathbb{R}^{n}$, $g \in \mathbb{R}^{m}$ and $h \in \mathbb{R}^{l}$ for each constraint, the Lagrangian of~\eqref{eq:mainProb} reads\footnote{Note that being probability measures, the term $-\sum_{i,j}P_{i,j}+Q_{i,j}$ is zero, but we have included this as it simplifies some of the forthcoming expressions.}
\begin{eqnarray}
  \mathcal{E}(P,f,g,h)
  &:=& \beta \sum_{i,j} P_{i,j} C_{i,j} 
    + \varepsilon\, \left(\mathrm{KL}(P \mid\mid Q) -\sum_{i,j}P_{i,j}+Q_{i,j}\right)- \sum_{i}f_i \left(\sum_{j}P_{i,j}-a_i\right) \nonumber\\
    &-& \sum_{j}g_j \left(\sum_{i}P_{i,j}-b_i\right) \label{eq:Lagrangian}
    - \sum_{k}h_k \left(\sum_{(i,j)\in D_k}P_{i,j}-c_k\right)
\end{eqnarray}
First-order conditions then yield
\begin{equation*}
  \frac{\partial \mathcal{E}(P,f,g,h)}{\partial P_{i,j}}
  = \beta C_{i,j} + \varepsilon \log(\frac{P_{i,j}}{Q_{i,j}}) - f_i - g_j - h_k= 0,
\end{equation*}
where $k$ is such that $(i,j)\in D_{k}$.
The optimal coupling $P$ therefore has the form
$$
  P_{i,j}
  = Q_{i,j}e^{(f_i+g_j+h_k-\beta C_{i,j})/\varepsilon}.
$$
Plugging this into~\eqref{eq:Lagrangian} and cancelling relevant expressions, we get the result of the Theorem\footnote{In the $(i,j)$ summation we have made explicit the dependence of $k$ on $(i,j)$ for clarity.}.
\end{proof}

\begin{corollary}\label{cor:dual}
The Shr\"odinger potentials $f^*,g^*,h^*$ that solve the concave optimisation problem in Theorem~\ref{thm:dual} satisfy
$$
L^{\varepsilon,\beta}_{C}(a,b,c)=\sum_{i}f^*_i a_i + \sum_{j}g^*_j b_j + \sum_{k}h^*_k c_k,
$$
and the joint probability reads
$$
  P_{i,j}
  = Q_{i,j}e^{(f^*_i+g^*_j+h^*_k-\beta C_{i,j})/\varepsilon}.
$$
\end{corollary}

\section{Implementation details}\label{sec:impl}

In this section we discuss the practical implementation details of solving Theorem~\ref{thm:dual}.

\subsection{Discretisation}
In applications $n$ and $m$ will be large and we do not want to solve for a substantial number of parameters.
We therefore approximate (similar in spirit to~\cite{Guyon2024} and~\cite{GuyonBourgey2024})
\begin{align}
L_{C}^{\varepsilon,\beta}(a,b,c)
&\approx \max_{\theta} \Bigg\{
  m_0 + m_1\,\mathbb{E}[X] + m_2\,\mathbb{E}[Y]
  + \sum_{r\in\mathcal{K}_X} q^X_r C^X(r)
  + \sum_{r\in\mathcal{K}_Y} q^Y_r C^Y(r) \notag \\
 &\qquad  + \sum_{r\in\mathcal{K}_Z} q^Z_r C^Z(r) 
- \varepsilon \sum_{i,j} Q_{ij}\,
  \exp\!\Bigg(
    \frac{1}{\varepsilon}\Big[
      m_0 + m_1 x_i + m_2 y_j
      + \sum_{r\in\mathcal{K}_X} q^X_r (x_i - r)^{+} \notag \\
&\qquad\qquad\qquad
     + \sum_{r\in\mathcal{K}_Y} q^Y_r (y_j - r)^{+}  
     + \sum_{r\in\mathcal{K}_Z} q^Z_r (x_i - y_j - r)^{+}
      - \beta C_{ij}
    \Big]
  \Bigg)
  + \varepsilon
\Bigg\},
\end{align}
where
$$
\theta:=\{m_0,m_1,m_2\}\cup \{q^X_r\}_{r\in\mathcal{K}_X} \cup \{q^Y_r\}_{r\in\mathcal{K}_Y}  \cup \{q^Z_r\}_{r\in\mathcal{K}_Z},
$$
$C^{X},C^{Y},C^{Z}$ are target call option prices on $X,Y,Z$ and $\mathcal{K}_{X},\mathcal{K}_{Y},\mathcal{K}_{Z}$ are the corresponding strike sets (typically 20 points per set spanning +/- 3 standard deviations around at-the-money).
The Jacobian is then given explicitly by
\begin{align*}
\nabla_{\theta} L
&= \begin{pmatrix}
\frac{\partial L}{\partial m_0} \\[4pt]
\frac{\partial L}{\partial m_1} \\[4pt]
\frac{\partial L}{\partial m_2} \\[4pt]
\{\frac{\partial L}{\partial q^X_r}\}_{r\in\mathcal{K}_X} \\[4pt]
\{\frac{\partial L}{\partial q^Y_r}\}_{r\in\mathcal{K}_Y} \\[4pt]
\{\frac{\partial L}{\partial q^Z_r}\}_{r\in\mathcal{K}_Z}
\end{pmatrix}
= \begin{pmatrix}
1 - \sum_{i,j} W_{ij} \\[4pt]
\mathbb{E}[X] - \sum_{i,j} W_{ij}\, x_i \\[4pt]
\mathbb{E}[Y] - \sum_{i,j} W_{ij}\, y_j \\[4pt]
\{C^X(r) - \sum_{i,j} W_{ij}(x_i - r)^{+}\}_{r\in\mathcal{K}_X} \\[4pt]
\{C^Y(r) - \sum_{i,j} W_{ij}(y_j - r)^{+}\}_{r\in\mathcal{K}_Y} \\[4pt]
\{C^Z(r) - \sum_{i,j} W_{ij}(x_i - y_j - r)^{+}\}_{r\in\mathcal{K}_Z}
\end{pmatrix},
\end{align*}
where
\begin{align*}
W_{ij}
&:= Q_{ij}\, \exp\!\left(\frac{A_{ij}(\theta)}{\varepsilon}\right), \\[8pt]
A_{ij}(\theta)
&:= m_0 + m_1 x_i + m_2 y_j
   + \sum_{r\in\mathcal{K}_X} q^X_r (x_i - r)^{+}
   + \sum_{r\in\mathcal{K}_Y} q^Y_r (y_j - r)^{+} \\
&\quad + \sum_{r\in\mathcal{K}_Z} q^Z_r (x_i - y_j - r)^{+}
   - \beta C_{ij}.
\end{align*}
The Hessian has entries
\begin{equation*}
\frac{\partial^2 L}{\partial \theta_k \, \partial \theta_\ell}
= -\frac{1}{\varepsilon} \sum_{i,j} W_{ij}\, \phi^{(\theta_k)}_{ij}\, \phi^{(\theta_\ell)}_{ij},
\end{equation*}
where the derivatives $\phi^{(\theta_k)}_{ij} := \partial_{\theta_k} A_{ij}$ are given by
\begin{align*}
\phi^{(m_0)}_{ij} &= 1,
\qquad
\phi^{(m_1)}_{ij} = x_i,
\qquad
\phi^{(m_2)}_{ij} = y_j, \\
\phi^{(q^X_r)}_{ij} &= (x_i - r)^{+},
\qquad
\phi^{(q^Y_r)}_{ij} = (y_j - r)^{+},
\qquad
\phi^{(q^Z_r)}_{ij} = (x_i - y_j - r)^{+}.
\end{align*}
We solve this optimisation problem using the Jacobian and Hessian and a gradient based method (e.g. a damped Newton method) and the initial condition is given by $\theta\equiv 0$ (this just starts with the prior joint density as an initial guess).

Let $(\xi^{X}_i)_{i=1}^{n}$ and $(\xi^{Y}_j)_{j=1}^{m}$ be node points from a $n$ and $m$ point Gauss-Legendre scheme (in our case we set $n=m=200$) and let $w_i^{X}$ and $w_j^{Y}$ be the corresponding weights respectively.
Denote the associated CDF's of $X$ and $Y$ as $M^{X}$ and $M^{Y}$.
For a small $q>0$ we choose
$x_{low}=\left(M^{X}\right)^{-1}(q)$, $y_{low}=\left(M^{Y}\right)^{-1}(q)$, $x_{high}=\left(M^{X}\right)^{-1}(1-q)$ and $y_{high}=\left(M^{Y}\right)^{-1}(1-q)$.
Then we set
$$
Q_{i,j}=c\left(M^{X}(x_i), M^{Y}(y_j)\right)f^{X}(x_i)  f^{Y}(y_j)\widetilde{w}_i^{X} \widetilde{w}_j^{Y},
$$
where $c=\partial^2_{uv}C(u,v)$ is a given copula density, $f^{X}, f^{Y}$ are the associated marginal densities 
$$
\widetilde{w}_i^{X} =\frac{x_{high}-x_{low}}{2}w_i^{X},\qquad \widetilde{w}_j^{Y} =\frac{y_{high}-y_{low}}{2}w_j^{Y},
$$
and
$$
x_i=\frac{x_{high}-x_{low}}{2}\xi^{X}_i+\frac{x_{high}+x_{low}}{2},\qquad y_j=\frac{y_{high}-y_{low}}{2}\xi^{Y}_j+\frac{y_{high}+y_{low}}{2}.
$$
In this paper we consider the Gaussian and t-copula cases.
More specifically, in the Gaussian copula case we have that
$$
c(u,v;\rho)=\frac{1}{\sqrt{1-\rho^2}}\exp\left(\frac{2\rho\Phi^{-1}(u)\Phi^{-1}(v)-\rho^2(\Phi^{-1}(u)^2+\Phi^{-1}(v)^2)}{2(1-\rho^2)}\right),
$$
and in the t-copula case we have that
\begin{align*}
c(u,v; \rho, \nu) = \frac{\Gamma(\frac{\nu+2}{2})\Gamma(\frac{\nu}{2})}{[\Gamma(\frac{\nu+1}{2})]^2} \frac{1}{\sqrt{1-\rho^2}} \left(1 + \frac{q_\rho}{\nu}\right)^{-\frac{\nu+2}{2}} %\\
\times \prod_{i=1,2} \left(1 + \frac{t_i^2}{\nu}\right)^{\frac{\nu+1}{2}},
\end{align*}
where $q_\rho = \frac{t_1^2 - 2\rho t_1 t_2 + t_2^2}{1-\rho^2}$ and $t_1= T_\nu^{-1}(u),t_2= T_\nu^{-1}(v)$ are t-distribution quantiles.
The reason for choosing a $t$-copula is that we can use the degrees of freedom ($\nu$) to control the heaviness of the tails and the strength of tail dependence and assess how this affects exotic prices via the prior joint probabilities. 
One can use any reasonable prior joint pdf in place of the ones above.   

\subsection{Adaptive epsilon scheduling}\label{sec:adaptive}

As $\varepsilon\downarrow 0$, the solution to~\eqref{eq:mainProb} converges to the non-entropic optimal transport solution (the constraint set $\mathcal{M}$ is closed and
bounded). 
However, for small $\varepsilon$, the optimisation becomes increasingly ill-conditioned. Exponential terms can become extremely large or small, leading to numerical instability in gradient and Hessian computations. This is a well-known challenge in the entropic optimal transport literature.

To enhance numerical stability and convergence for small $\varepsilon$ (typically $\varepsilon<0.01$), we employ adaptive epsilon scheduling (or annealing; see, e.g.,~\cite{schmitzer2019}). This approach begins with a larger value of $\varepsilon$, solves the corresponding problem, and uses the solution as a warm start for progressively smaller values of $\varepsilon$. Large $\varepsilon$ values correspond to more regularized problems that converge rapidly, and warm-starting accelerates convergence as $\varepsilon$ decreases. We define an adaptive schedule $\{\varepsilon_k\}_{k=1}^K$ as follows:
\begin{equation*}
\varepsilon_k = \varepsilon_{\text{start}} \cdot \left(\frac{\varepsilon_{\text{target}}}{\varepsilon_{\text{start}}}\right)^{\frac{k-1}{K-1}},
\quad
\varepsilon_{\text{start}} = \max(0.05, 10\varepsilon_{\text{target}}),
\quad
K = \max\left(5, \left\lfloor -2\log\left(\frac{\varepsilon_{\text{target}}}{\varepsilon_{\text{start}}}\right) \right\rfloor\right).
\end{equation*}
This approach is faster than directly solving the linear (or quadratic) program for the non-entropic optimal transport problem, as done in~\cite{Piterbarg2011}\footnote{For complexity analyses of entropic optimal transport (including Sinkhorn and gradient\text{-}based solvers) and for pointers to complexity results on the unregularized optimal transport problem (e.g., network simplex and interior\text{-}point methods), see~\cite{dvurechensky2018} and the references therein.}.

\subsection{Explicit translation of problem to interest rate markets}\label{sec: IR}

Fix an expiry $T>0$ and let $S(t,T,T^{'})$ be a forward swap rate observed at time $t\leq T$ for the period $[T,T']$.
Then in our case we set $X=S(T,T,T_1)$ and $Y=S(T,T,T_0)$ for $T_1>T_0$.
All expectations and PDF's/CDF's are taken under the $T$-forward measure.
We suppose that we are given a market of CMS options on $X$ and $Y$ and a market of CMS spread options on the difference $X-Y$.

\section{Numerical examples}\label{sec:numerics}

We first consider a 2s20s CMS pair with 5-year expiry. 
Figures~\ref{fig:2CMSBaseGC20s2s2025}--\ref{fig:20y2yBaseGC2025} demonstrate the calibration accuracy for the ``Base GC'' specification ($\varepsilon=1$, $\beta=0$, Gaussian copula prior), while Figures~\ref{fig:2CMSLowerGC2025}--\ref{fig:20y2yLowerGC2025} show results for the ``Lower GC'' specification ($\varepsilon=0.001$, $\beta=-1$ and the payoff is an option on $X+Y$ with at-the-money strike $\mathbb{E}[X+Y]$).\footnote{As discussed in Section~\ref{sec: IR}, all expectations are taken under the $T$-forward measure.}

CMS and spread option volatilities are implied from prices using the Bachelier model, with forwards set to the respective CMS forward and CMS spread forward rates. The calibration employs 20 strikes spanning $\pm 3$ standard deviations around at-the-money for each of the two CMS rates and the spread rate. Calibration quality is excellent across all examples; all results presented in this section achieve similar accuracy.

\begin{figure}[H]
    \centering
    \includegraphics[width=0.8\textwidth]{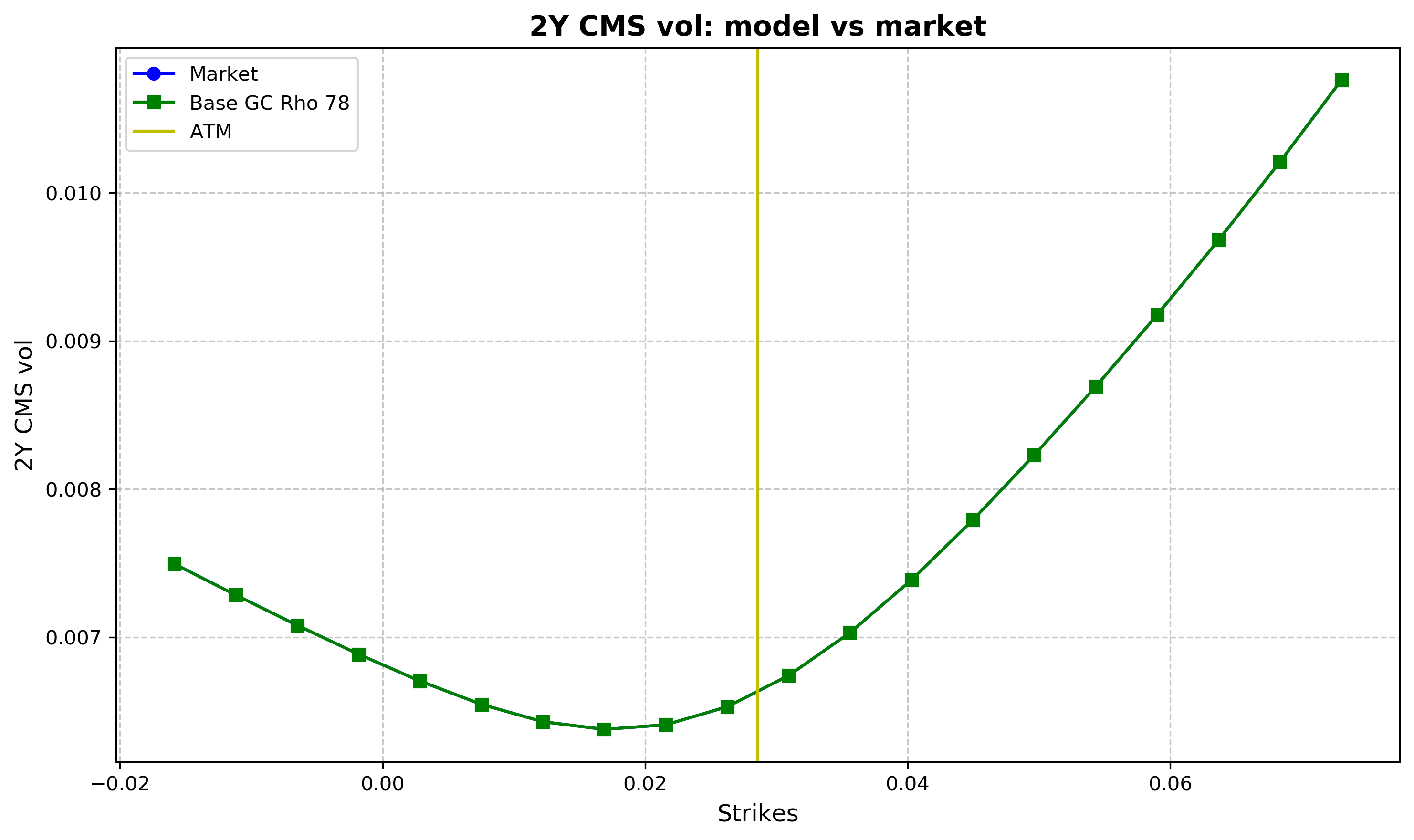}
    \caption{Calibration accuracy for 2Y CMS options for ``Base GC' with $\rho=78\%$.}
    \label{fig:2CMSBaseGC20s2s2025}
\end{figure}
\begin{figure}[H]
    \centering
    \includegraphics[width=0.8\textwidth]{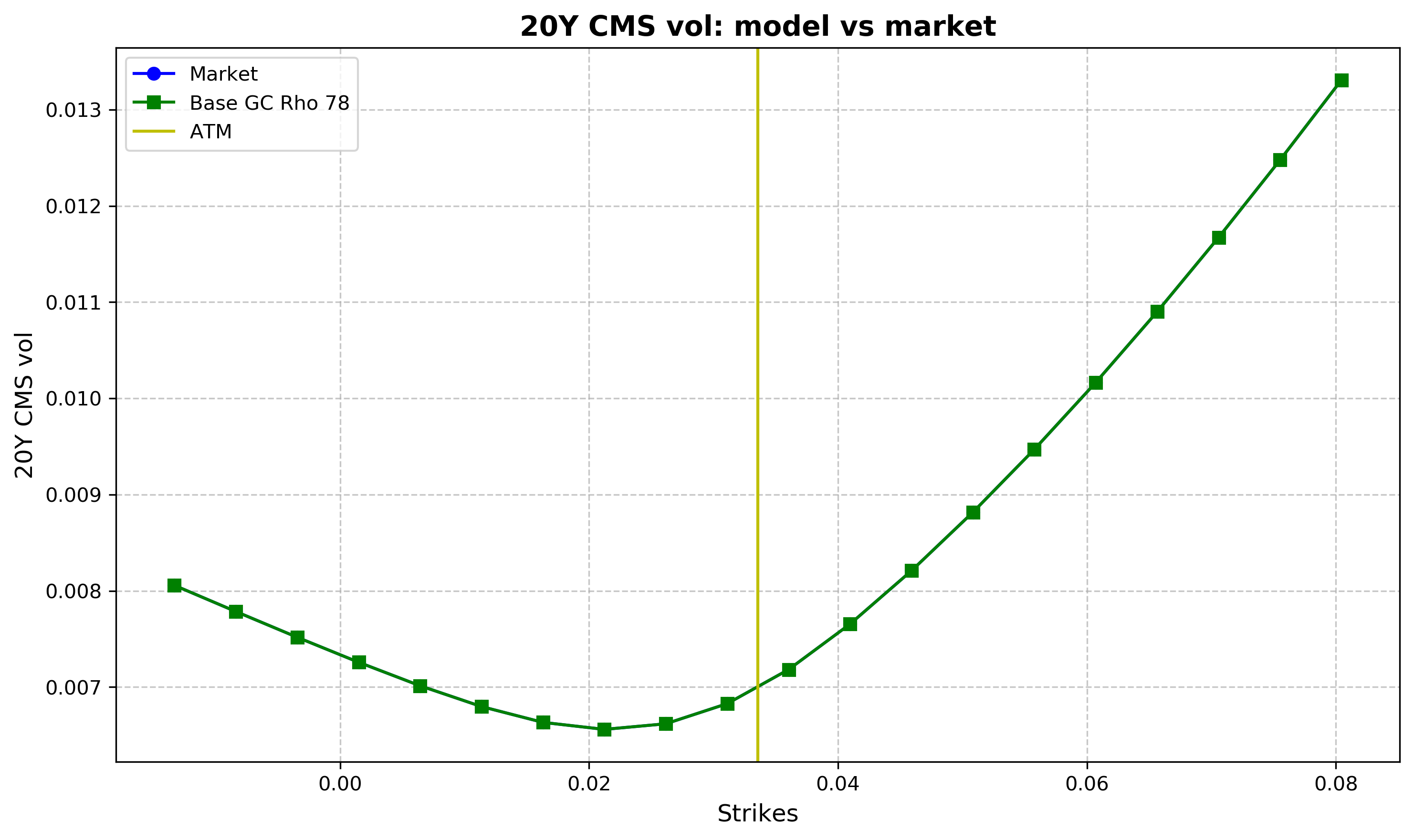}
    \caption{Calibration accuracy for 20Y CMS options for ``Base GC' with $\rho=78\%$.}
    \label{fig:20CMSBaseGC2025}
\end{figure}
\begin{figure}[H]
    \centering
    \includegraphics[width=0.8\textwidth]{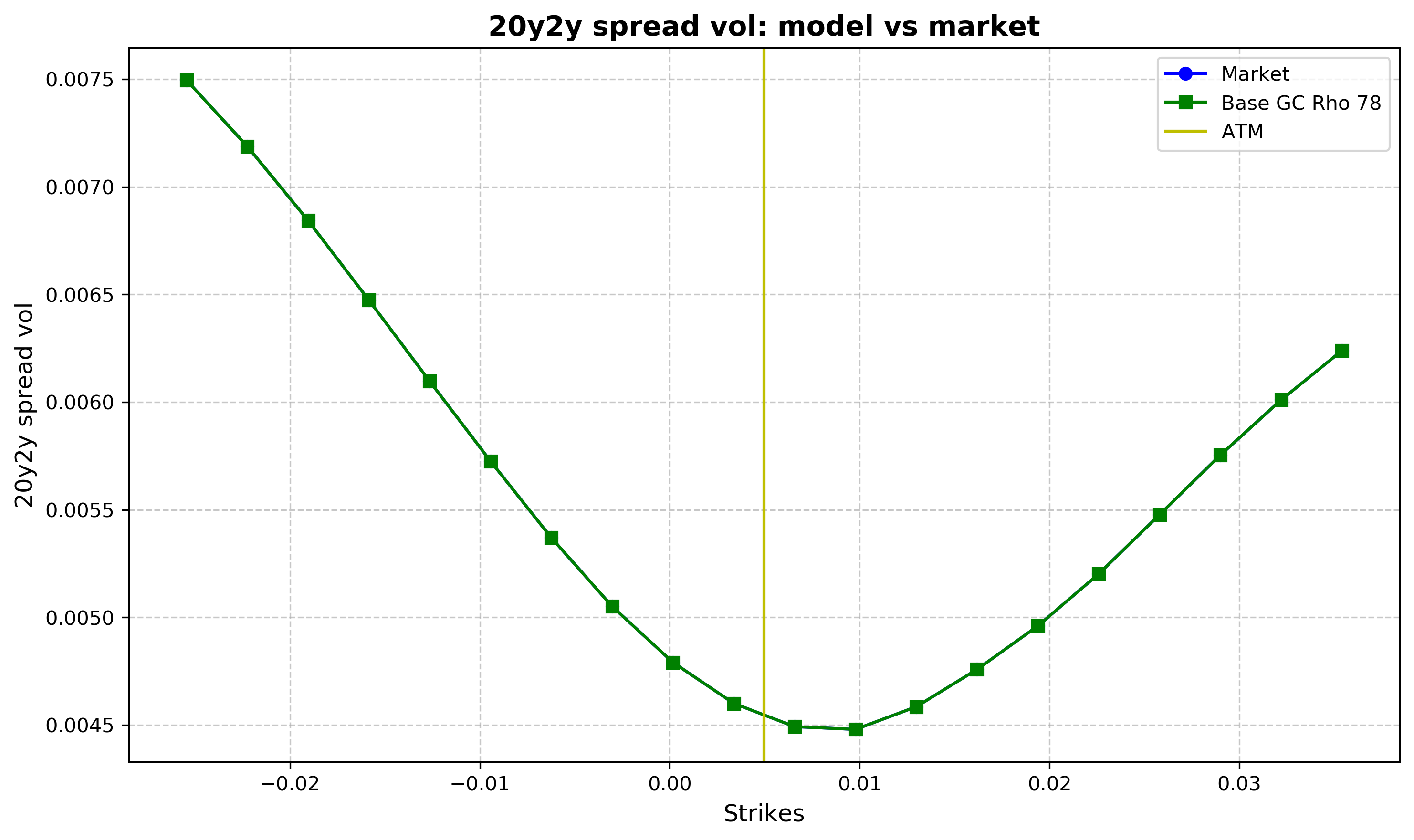}
    \caption{Calibration accuracy for 2s20s CMS spread options for ``Base GC' with $\rho=78\%$.}
    \label{fig:20y2yBaseGC2025}
\end{figure}

\begin{figure}[H]
    \centering
    \includegraphics[width=0.8\textwidth]{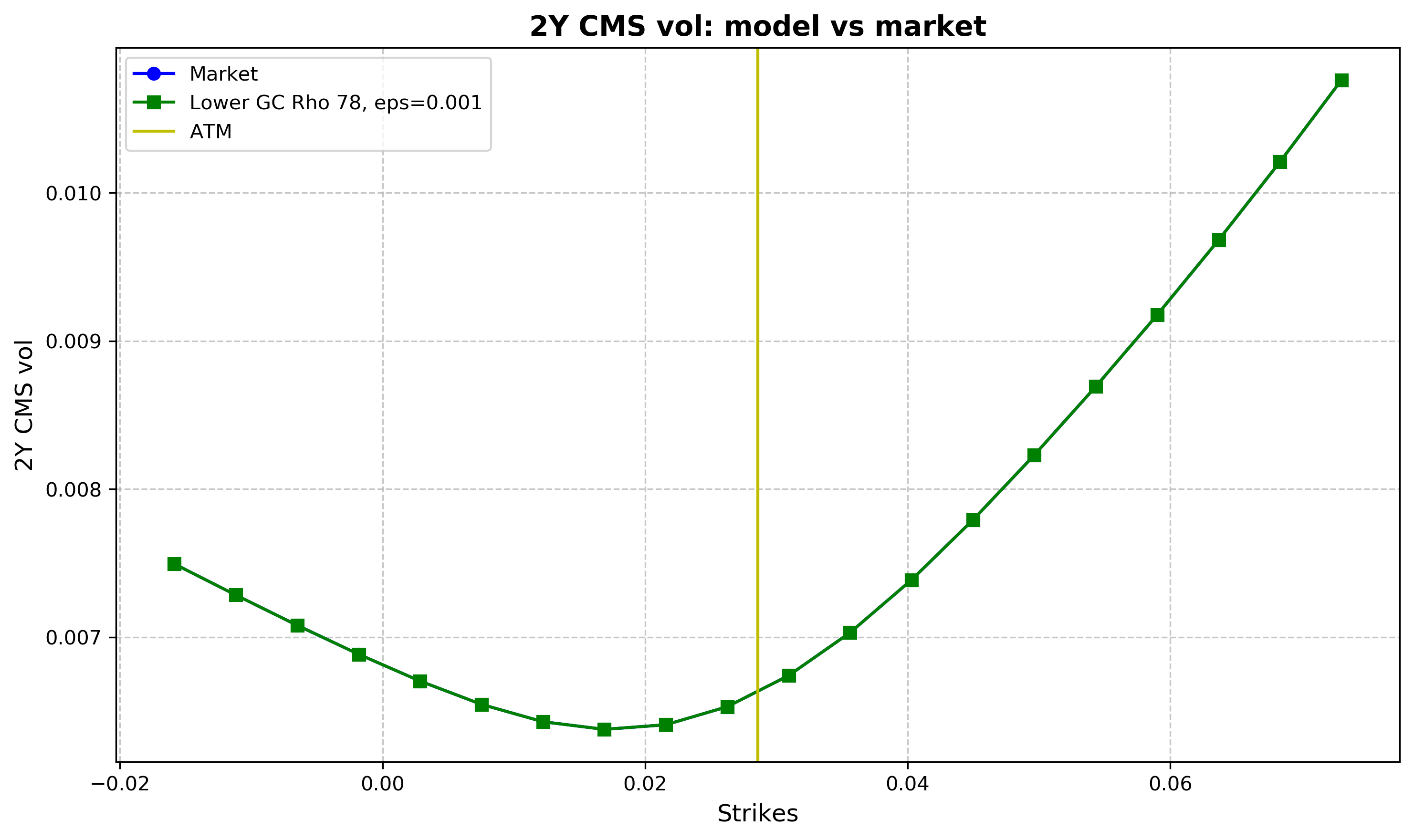}
    \caption{Calibration accuracy for 2Y CMS options for ``Lower GC' with $\rho=78\%$ and $\varepsilon=0.001$ (payoff is an ATM option on $X+Y$ as discussed in the text).}
    \label{fig:2CMSLowerGC2025}
\end{figure}
\begin{figure}[H]
    \centering
    \includegraphics[width=0.8\textwidth]{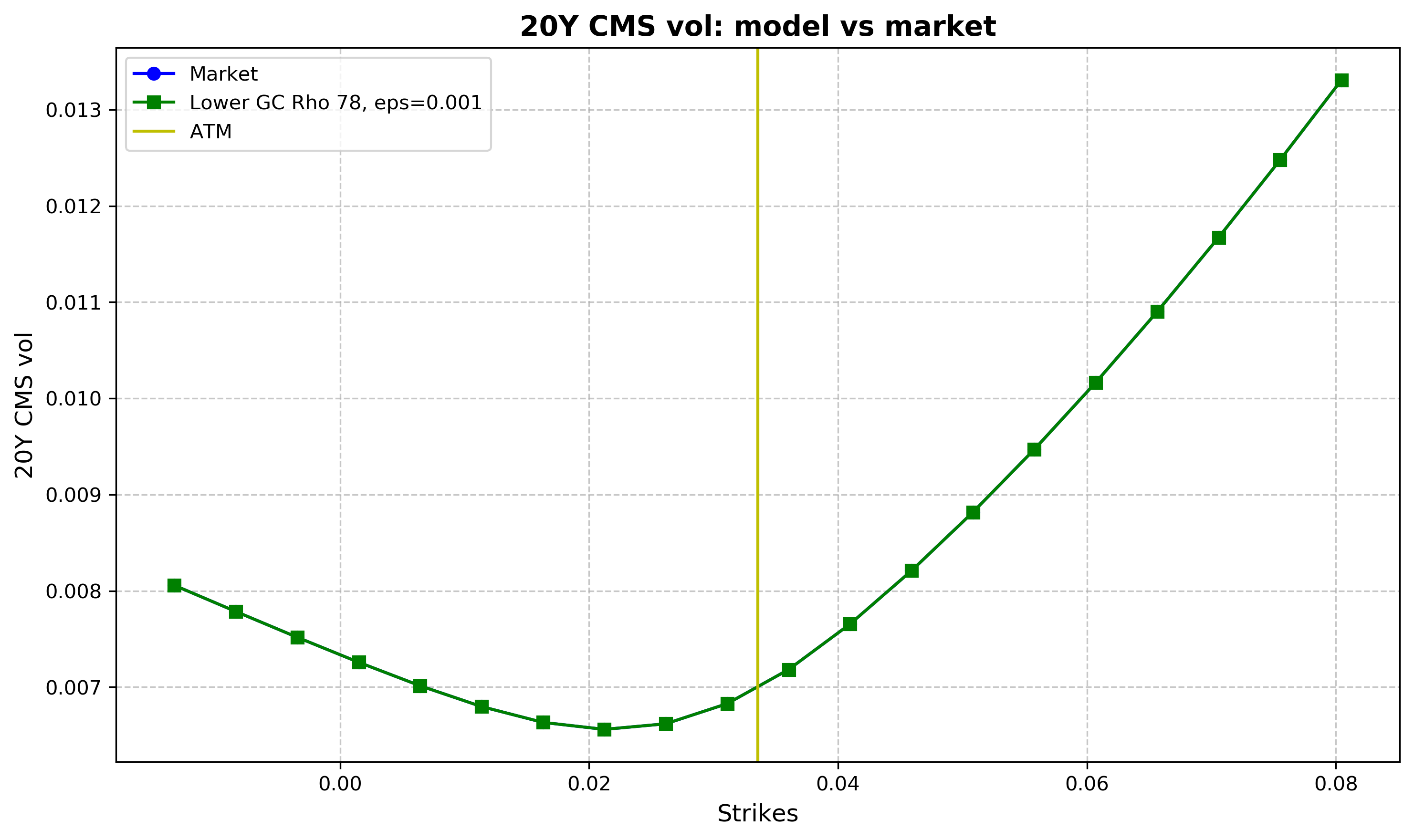}
    \caption{Calibration accuracy for 20Y CMS options for ``Lower GC' with $\rho=78\%$ and $\varepsilon=0.001$ (payoff is an ATM option on $X+Y$ as discussed in the text).}
    \label{fig:20CMSLowerGC2025}
\end{figure}
\begin{figure}[H]
    \centering
    \includegraphics[width=0.8\textwidth]{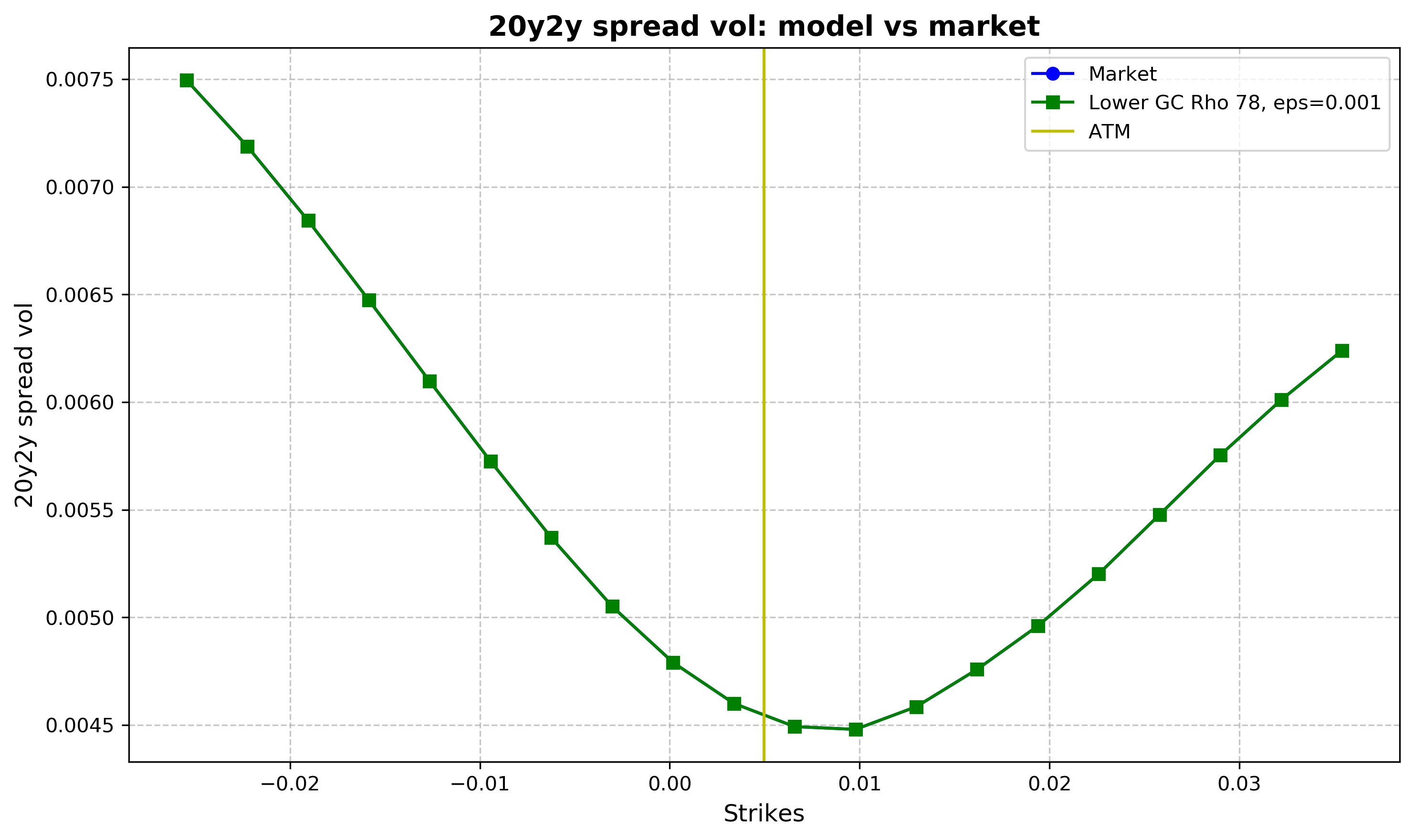}
    \caption{Calibration accuracy for 2s20s CMS spread options for ``Lower GC' with $\rho=78\%$ and $\varepsilon=0.001$ (payoff is an ATM option on $X+Y$ as discussed in the text).}
    \label{fig:20y2yLowerGC2025}
\end{figure}

We now present results for an option on $X+Y$ using a Gaussian copula prior with varying levels of $\varepsilon$, displayed in Figure~\ref{fig:rho_vs_exotic_strikes_20y2y_5Y_2_2025}. The expiry is again 5 years, and results are expressed in terms of implied Gaussian copula correlation (defined as the correlation parameter needed in a Gaussian copula model with marginal distributions of $X$ and $Y$ to reprice the option on $X+Y$).

As $\varepsilon\downarrow 0$, the solution to~\eqref{eq:mainProb} approaches the non-entropic optimal transport bounds for the payoff (see Section~\ref{sec:adaptive}). These bounds are displayed in the figure for $\varepsilon=0.001$.\footnote{Reducing $\varepsilon$ further produced no appreciable change in the bounds. Also, note that the prior becomes irrelevant for this level of $\varepsilon$. i.e. Changing to a t-copula prior with the same $\varepsilon$ yields essentially the same bounds.} We also present results for $\varepsilon=0.01$, which we believe strikes a good balance between anchoring to an economically meaningful joint density (such as a Gaussian copula) and capturing the model risk inherent in the product. 
Although the non-entropic optimal transport bounds encompass all joint distributions satisfying the marginal and spread constraints, some distributions achieving these bounds can appear financially implausible. Piterbarg~\cite{Piterbarg2011} observed this phenomenon and introduced a quadratic penalty term in the linear program to discourage such solutions. 

We further note that the bounds can be substantial and tend to widen as spread volatility increases. Figure~\ref{fig:rho_vs_exotic_strikes_30y2y_5Y_2_2022} illustrates this for a 5-year expiry 30y2y tenor pair on a business date when spread volatilities were significantly elevated.
In Figure~\ref{fig:rho_vs_exotic_strikes_20y2y_5Y_3_2025}, we examine the impact of the prior specification by employing a t-copula with varying degrees of freedom. This allows us to assess how increased tail dependence in the prior affects exotic prices and provides further insight into the model risk.

Finally, we note that in a small number of cases, the methodology failed to construct a joint distribution. Upon inspection, these failures arose from inconsistencies between out-of-the-money options on the CMS rates and the spread, strongly suggesting potential arbitrage and the non-existence of a joint distribution satisfying all constraints.

\begin{figure}[H]
    \centering
    \includegraphics[width=0.8\textwidth]{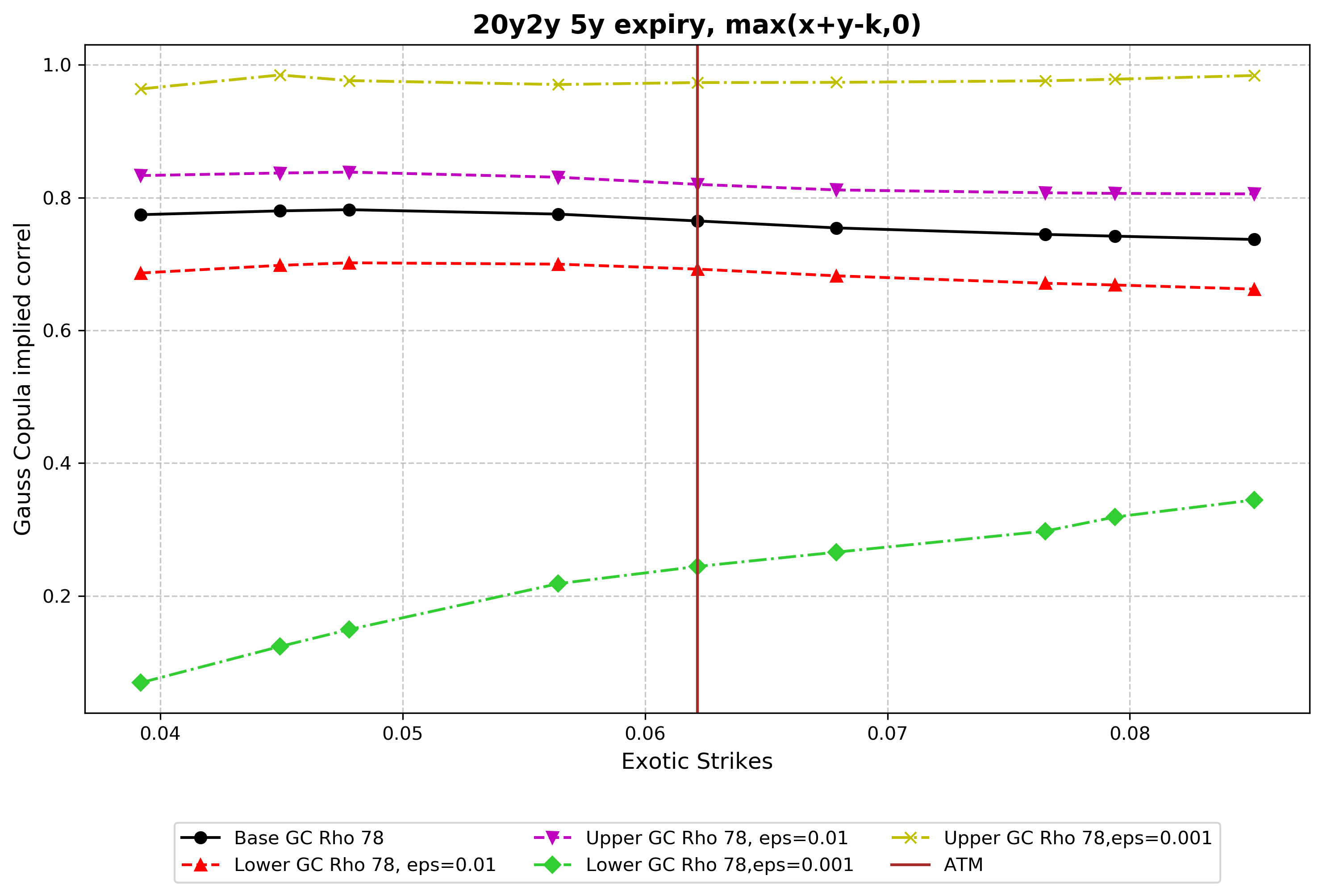}
    \caption{No-arbitrage bounds for the exotic given different $\varepsilon$.}
    \label{fig:rho_vs_exotic_strikes_20y2y_5Y_2_2025}
\end{figure}
\begin{figure}[H]
    \centering
    \includegraphics[width=0.8\textwidth]{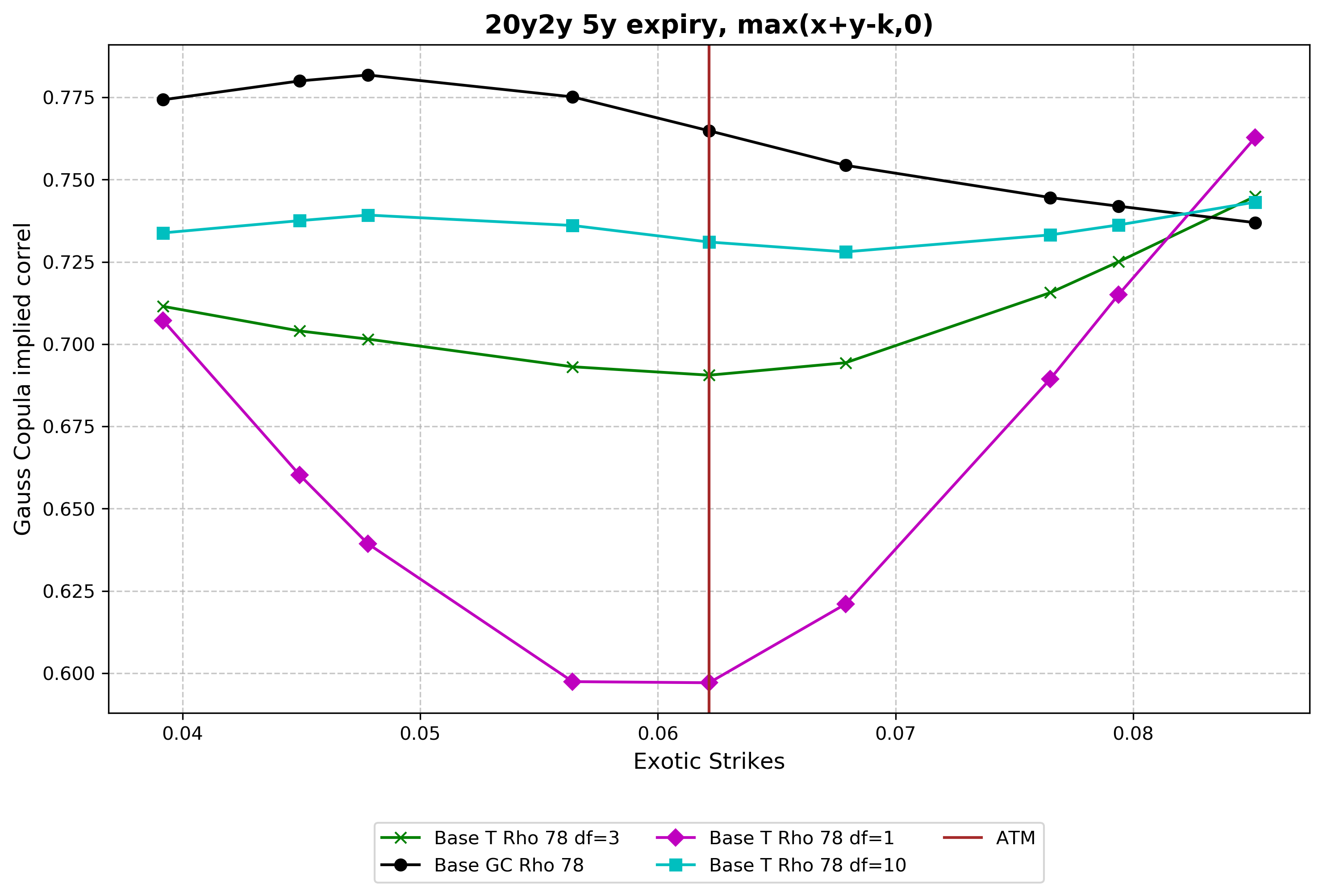}
    \caption{Price comparisons using different prior joint probabilities.}
    \label{fig:rho_vs_exotic_strikes_20y2y_5Y_3_2025}
\end{figure}

\begin{figure}[H]
    \centering
    \includegraphics[width=0.8\textwidth]{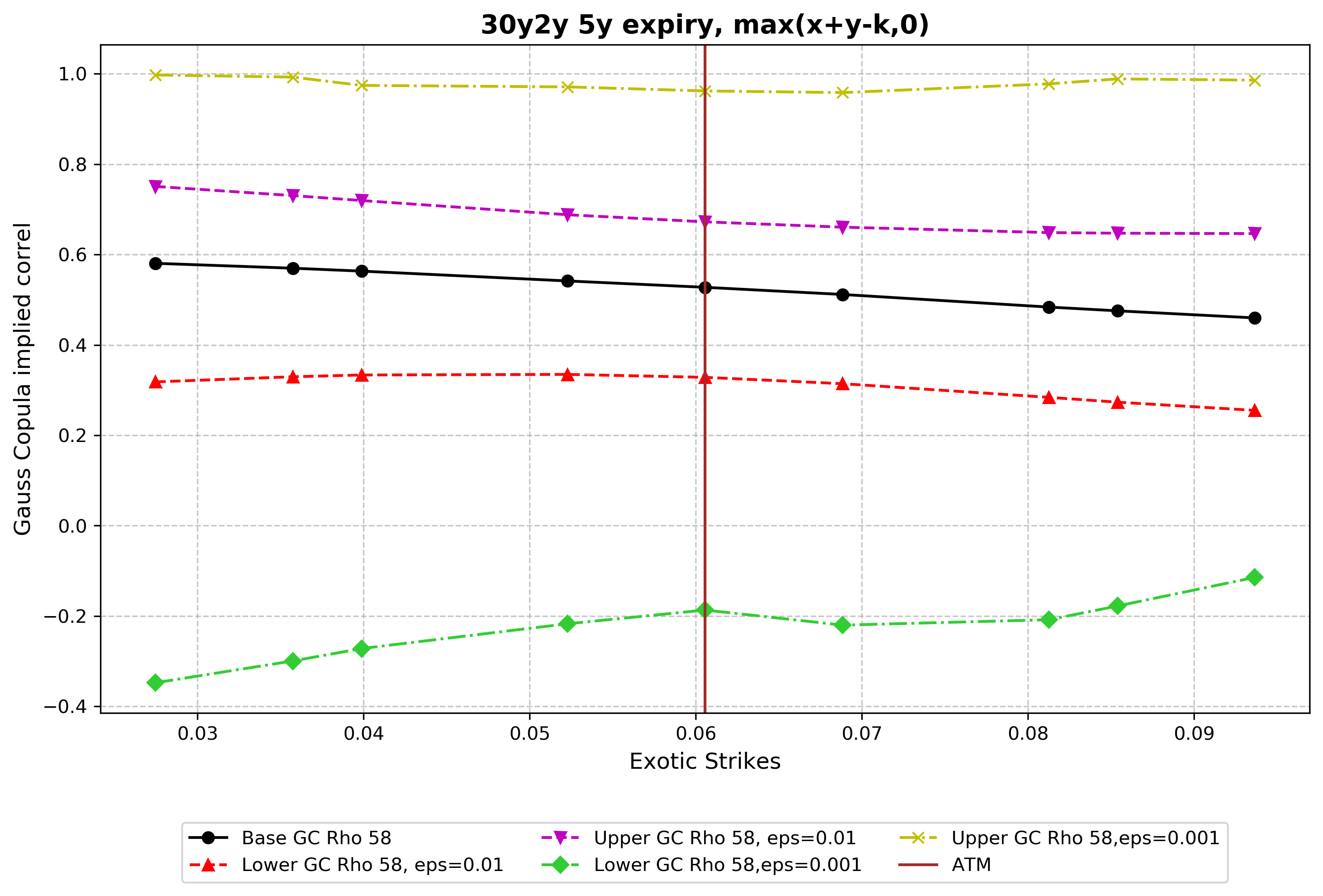}
    \caption{No-arbitrage bounds for the exotic given different $\varepsilon$ for a business date when spread volatilities were significantly elevated.}
    \label{fig:rho_vs_exotic_strikes_30y2y_5Y_2_2022}
\end{figure}

\noindent\textit{E-mail address:} \texttt{patrick.roome@yahoo.com}

\end{document}